\documentclass[journal,10pt,doublecolumn]{IEEEtran}
\usepackage{amsmath}
\usepackage{graphicx}
\usepackage{cite}
\usepackage{amsfonts}
\usepackage{amsthm}
\usepackage{algorithm}
\usepackage{algorithmic}
\usepackage{amssymb}
\usepackage{color}
\ifCLASSOPTIONcompsoc
\usepackage[caption=false,font=normalsize,labelfont=sf,textfont=sf]{subfig}
\else
\usepackage[caption=false,font=footnotesize]{subfig}
\fi

\theoremstyle{definition}  \newtheorem{lemma}{Lemma}
\theoremstyle{definition}  
\theoremstyle{definition}

\begin{document}

\title{Robust Trajectory and Transmit Power Design for Secure UAV Communications}

\author{Miao~Cui, Guangchi~Zhang,~\IEEEmembership{Member,~IEEE,} Qingqing~Wu,~\IEEEmembership{Member,~IEEE,}~and~Derrick~Wing~Kwan~Ng,~\IEEEmembership{Senior~Member,~IEEE}
	\thanks{M. Cui and G. Zhang are with the School of Information Engineering, Guangdong University of Technology, Guangzhou, China (email: \{cuimiao, gczhang\}@gdut.edu.cn). Q. Wu is with the Department of Electrical and Computer Engineering, National University of Singapore, Singapore (email: elewuqq@nus.edu.sg). D. W. K. Ng is with the School of Electrical Engineering and Telecommunications, University of New South Wales, Australia (email: w.k.ng@unsw.edu.au). G. Zhang is the corresponding author. } }

% make the title area
\maketitle

\begin{abstract}
Unmanned aerial vehicles (UAVs) are anticipated to be widely deployed in future wireless communications, due to their advantages of high mobility and easy deployment. However, the broadcast nature of air-to-ground line-of-sight wireless channels brings a new challenge to the information security of UAV-ground communication. This paper tackles such a challenge in the physical layer by exploiting the mobility of UAV via its trajectory design. We consider a UAV-ground communication system with multiple potential eavesdroppers on the ground, where the information on the locations of the eavesdroppers is imperfect. We formulate an optimization problem which maximizes the average worst-case secrecy rate of the system by jointly designing the robust trajectory and transmit power of the UAV over a given flight duration. The non-convexity of the optimization problem and the imperfect location information of the eavesdroppers make the problem difficult to be solved optimally. We propose an iterative suboptimal algorithm to solve this problem efficiently by applying the block coordinate descent method, $\mathcal{S}$-procedure, and successive convex optimization method. Simulation results show that the proposed algorithm can improve the average worst-case secrecy rate significantly, as compared to two other benchmark algorithms without robust design.
\end{abstract}

%\begin{IEEEkeywords}
%UAV communications, physical layer security, secrecy rate maximization, trajectory optimization, power control, robust design.
%\end{IEEEkeywords}

%we need to discuss more on our strong points and the difference of this work with previous work such as (robust and multiple Evs, Llike if we dont do this design, what will happen to provide more intutions.)

\vspace{-0.4cm}

\section{Introduction} \vspace{-0.1cm}
Due to the advantages of high mobility and flexibility, unmanned aerial vehicles (UAVs) have found interesting applications in wireless communications \cite{zeng2016wireless,qingqing18UAVmagazine,Mozaffari2015,Yang2018,Bor2016}. Line-of-sight (LoS) channels usually exist between UAVs and ground nodes in UAV wireless communication systems \cite{Lin2017}.
  This has also inspired a proliferation of studies recently on the new research paradigm of jointly optimizing the UAV trajectory design and communication resource allocation, for e.g.  multiple access channel (MAC) and broadcast channel (BC) \cite{JR:wu2017_capacity,JR:wu2017_ofdm}, interference channel (IFC)\cite{WuGC2017}, and  wiretap channel \cite{Zhang2017GC,Lian2018wcl}.  In particular, as shown in  \cite{ JR:wu2017_capacity} and  \cite{JR:wu2017_ofdm}, significant communication throughput gains can be achieved by mobile UAVs over static UAVs/fixed terrestrial BSs by exploiting the new design degree of freedom via UAV trajectory optimization, especially for delay-tolerant applications.  In \cite{WuGC2017},  a joint UAV trajectory, user association, and power control optimization framework is proposed for cooperative  multi-UAV enabled wireless networks. However, legitimate UAV-ground communications are more prone to be intercepted by potential eavesdroppers on the ground, as compared to terrestrial wireless communication systems, which gives rise to a new security challenge. Although security can be conventionally handled by using cryptographic methods adopted in the higher communication protocol layers, physical layer security is now emerging as a promising alternative technology to realize secrecy in wireless communication \cite{Gopala2008}. One widely adopted performance metric in the physical layer security design is the so-called secrecy rate \cite{QiangLi2015}, at which confidential information can be reliably conveyed. For secure UAV communications, a joint UAV trajectory and transmit power control design framework has been proposed in \cite{Zhang2017GC}, where the average secrecy rate is maximized by proactively enhancing the legitimate link and degrading the eavesdropping link via UAV trajectory design in addition to power adaptation. However, the location of the eavesdropper is assumed to be perfectly known in \cite{Zhang2017GC}, which is overly optimistic. {In practice, although the UAV can estimate the location of a potential eavesdropper by applying an camera or synthetic aperture radar \cite{Li2015SAR},} the eavesdropper may remain silent to hide its existence and thus the location estimation is expected suffering from errors. As a result, existing security-enabled techniques based on the assumption of perfect location information of eavesdroppers may result in significant degradation on security performance. Moreover, there may be more than one eavesdroppers trying to intercept the legitimate UAV-ground communication in practice. In this scenario, the UAV transmitter needs to steer away from multiple eavesdroppers and at the same time approach its intended receiver as close as possible to enhance secrecy rate. Hence, designing the UAV trajectory in such a scenario is an interesting but challenging problem, which has not been addressed in \cite{Zhang2017GC}.

\begin{figure}[!t]
	\centering
	\includegraphics[width=0.8\columnwidth]{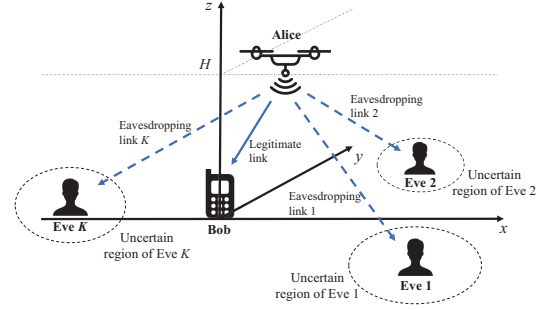}
	\caption{A UAV (Alice) communicates with a ground node (Bob) with $K$ potential eavesdroppers (Eves) on the ground.}  \label{Figphysical}
\end{figure}

In this paper, we consider secure legitimate UAV-ground communications via robust joint UAV trajectory and transmit power design in a practical scenario, where there are multiple eavesdroppers on the ground, as shown in Fig. \ref{Figphysical}. The UAV only knows the approximate regions in which the eavesdroppers are located while the exact locations of the eavesdroppers are unknown. { We aim to maximize the average worst-case secrecy rate over a given flight duration of the UAV, subject to its mobility constraints as well as its average and peak transmit power constraints. The main contributions are summarized as follows. \vspace{-0.1cm}
\begin{itemize}
	\item The considered problem is intractable and obtaining the globally optimal solution is difficult due to its non-convexity and semi-infinite numbers of constraints. To tackle the intractability, we propose an efficient suboptimal algorithm to solve this problem, based on the block coordinate descent method, $\mathcal{S}$-Procedure, and the successive convex optimization method.
	\item Since the proposed algorithm takes into account and provides robustness against the imperfect location information of multiple eavesdroppers, it is more suitable for practical applications, as compared to the existing work on secure UAV communications \cite{Zhang2017GC,Lian2018wcl}.
	\item Simulation results show that the proposed algorithm can improve the average worst-case secrecy rate significantly, compared to other benchmark schemes assuming perfect location information of the eavesdroppers or ignoring the eavesdroppers.
\end{itemize}
}

\vspace{-0.4cm}

\section{System Model And Problem Formulation} \vspace{-0.1cm}
{We consider a UAV-ground communication system, where $K$ eavesdroppers (Eves) on the ground try to intercept the legitimate communication from a UAV (Alice) to a ground node (Bob), as shown in Fig. \ref{Figphysical}.} We express locations in the three-dimensional Cartesian coordinate system. Without loss of generality, we assume that Bob locates at $(0,0,0)$, which is perfectly known by Alice. For $k \in \mathcal{K} \triangleq \{1,\ldots,K\}$, the exact location of Eve $k$, denoted by $(x_k,y_k,0)$ in meters (m), is not known, but its estimated location, denoted by $(x_{\text{E}_k}, y_{\text{E}_k}, 0)$ in m, is assumed to be known. The relation between the actual and the estimated $x$-$y$ coordinates of Eve $k$ is given by \vspace{-0.2cm}
\begin{equation}  \label{EquLocErr}
x_k = x_{\text{E}_k} + \Delta x_k, \; y_k = y_{\text{E}_k} + \Delta y_k, \vspace{-0.2cm}
\end{equation}
respectively, where $\Delta x_k$ and $\Delta y_k$ denote estimation errors on $x_k$ and $y_k$, respectively, and satisfy the following condition  \vspace{-0.15cm}
\begin{equation}  \label{EquErrArea}
(\Delta x_k, \Delta y_k) \in \mathcal{E}_k \triangleq \{(\Delta x_k, \Delta y_k) | \Delta x_k^2 + \Delta y_k^2 \leq Q_k^2 \},  \vspace{-0.1cm}
\end{equation}
where $\mathcal{E}_k$ denotes a continuous set of possible errors. Thus, Eve $k$ can be regarded as locating in an uncertain circular region with center $(x_{\text{E}_k}, y_{\text{E}_k}, 0)$ and radius $Q_k$.

It is assumed that Alice flies at a constant altitude $H$ in m, which is specified for safety considerations such as building avoidance \cite{WuGC2017}. Thus, Alice's coordinate over time is denoted as $(x(t),y(t),H)$, $0 \leq t \leq T$, where $T$ in seconds (s) is its flight duration. {To facilitate trajectory design for Alice, we quantize the flight duration $T$ into $N$ sufficiently small time slots with equal length $d_t$. Since $d_t$ is small enough, Alice can be regarded as static within each slot. Thus, Alice's trajectory over the duration $T$ can be represented by a sequence $\{( x[n],y[n],H )\}_{n=1}^{N}$. We let $(x[0],y[0],H)$ and $(x[N+1],y[N+1],H)$ denote Alice's initial and final locations, respectively, and then write the mobility constraints of Alice as \vspace{-0.15cm}
\begin{equation}  \label{EquMobilityCon}
( x[n+1]-x[n] )^2 + ( y[n+1]-y[n] )^2 \leq (v_{\max} d_t)^2, \; \forall n, \vspace{-0.15cm}
\end{equation}
where $v_{\max}$ denotes the maximum speed of Alice.}

{The channel from Alice to Bob is assumed to be LoS channel \cite{WuGC2017,wu2018JSAC_ICC,wu2017joint_globecom}\footnote{Measurement results in \cite{Lin2017} show that the LoS channel model is a good approximation for UAV-ground communications in practice even if the UAV flies at a moderate altitude, e.g., $85$m.}. Thus, the power gain of the channel from Alice to Bob in slot $n$ is given by } \vspace{-0.2cm}
\begin{equation}
g_{\text{AB}}[n] = \beta_0 d_{\text{AB}}^{-2}[n] = \frac{ \beta_0 }{ x^2[n] + y^2[n] + H^2 },
\end{equation}
{where $\beta_0$ denotes the power gain of a channel with reference distance $d_0=1$m \cite{WuGC2017}, and $d_{\text{AB}}[n] = \sqrt{ x^2[n] + y^2[n] + H^2 }$ denotes the distance between Alice and Bob in slot $n$. Similarly, the channel from Alice to Eve $k$ can be assumed to be LoS channel, whose power gain in slot $n$ is given by } \vspace{-0.2cm}
\begin{equation}
g_{\text{AE}_k}[n] = \frac{ \beta_0 }{ (x[n]- x_k)^2 + (y[n] - y_k)^2 + H^2 }.  \vspace{-0.1cm}
\end{equation}

{ Let $P[n]$ denote the transmit power of Alice in slot $n$, and $\bar{P}$ and $P_{\text{peak}}$ denote the average power and peak power of Alice, respectively. Thus, we write the average and peak transmit power constraints of Alice as}  \vspace{-0.2cm}
\begin{subequations} \label{EquPowerCon}
\begin{align}
&\frac{1}{N} \sum_{n=1}^{N} P[n] \leq \bar{P},  \label{EquAvgPowCon} \\
&0 \leq P[n] \leq P_{\text{peak}}, \; \forall n.  \vspace{-0.1cm}
\end{align}
\end{subequations}
To ensure that \eqref{EquAvgPowCon} is a non-trivial constraint, we assume $\bar{P}<P_{\text{peak}}$. {Then, we can express the achievable rate from Alice to Bob in slot $n$ in bits/second/Hertz (bps/Hz) as}  \vspace{-0.2cm}
\begin{align}
R_{\text{AB}}[n] = & \log_2 \left(1 + \frac{ P[n] g_{\text{AB}}[n] }{ \sigma^2 }  \right)  \nonumber \\
= &   \log_2 \left( 1+ \frac{ \gamma_0 P[n] } {  x^2[n] + y^2[n] + H^2 }  \right), \label{EquRAB}  \vspace{-0.1cm}
\end{align}
{where $\gamma_0 =\beta_0 / \sigma^2$ and $\sigma^2$ is Gaussian noise power at the receiver. Similarly, we express the achievable rate from Alice to Eve $k$ in slot $n$ in bps/Hz as }  \vspace{-0.2cm}
\begin{equation}   \label{EquRAE}
R_{\text{AE}_{k}}[n]  = \log_2 \bigg( 1 + \frac{ \gamma_0 P[n] }{ (x[n]- x_k)^2 + (y[n] - y_k)^2 + H^2 } \bigg).
\end{equation}
With \eqref{EquRAB} and \eqref{EquRAE}, {the average worst-case secrecy rate from Alice to Bob over the flight duration $T$ in bps/Hz is \cite{QiangLi2015}}     \vspace{-0.2cm}
\begin{equation}   \label{EquSecrecyRate}
R_{\text{sec}} = \frac{1}{N} \sum_{n=1}^N \left[ R_{\text{AB}}[n] - \max_{ k\in \mathcal{K} } \max_{ (\Delta x_k, \Delta y_k) \in \mathcal{E}_k } R_{\text{AE}_k}[n] \right]^{+},
\end{equation}
where $[x]^+ \triangleq \max(x,0)$.

{ For secure the communication from Alice to Bob, we jointly design the trajectory and transmit power of Alice to maximize the average worst-case secrecy rate in \eqref{EquSecrecyRate} subject to its mobility and power constraints in \eqref{EquMobilityCon} and \eqref{EquPowerCon}. The optimization variables include Alice's trajectory and transmit power over $N$ time slots, which are denoted as $\mathbf{x} \triangleq \left[x[1], \ldots, x[N]\right]^{\dagger}$, $\mathbf{y} \triangleq \left[y[1], \ldots, y[N]\right]^{\dagger}$, and $\mathbf{P} \triangleq \left[P[1],\ldots,P[N] \right]^{\dagger}$, where $\dagger$ denotes the transpose operation. The problem is formulated as follows, where the constant term $1/N$ in \eqref{EquSecrecyRate} has been dropped, }  \vspace{-0.2cm}
\begin{align}
 \max_{ \mathbf{x}, \mathbf{y}, \mathbf{P} } & \;  \sum_{n=1}^N  \bigg[ R_{\text{AB}}[n]  - \max_{k\in \mathcal{K}} \max_{ (\Delta x_k, \Delta y_k) \in \mathcal{E}_k } R_{\text{AE}_k}[n] \bigg]^+  \label{EquOriginal} \\
\text{s.t.} \; & \; \eqref{EquMobilityCon}, \; \eqref{EquPowerCon}. \nonumber
\end{align}
{Problem \eqref{EquOriginal} is difficult to solve optimally because of the following reasons. First, the operator $[\cdot]^+$ introduces non-smoothness to the objective function. Second, the objective function is still not jointly concave with respect to $\mathbf{x}$, $\mathbf{y}$, and $\mathbf{P}$ even without $[\cdot]^+$. Third, the infinite number of possible $(\Delta x_k, \Delta y_k)$ makes \eqref{EquOriginal} an intractable semi-infinite optimization problem.} In the following section, we propose a computational efficient iterative suboptimal algorithm to solve problem \eqref{EquOriginal} approximately.

\vspace{-0.4cm}

\section{Proposed Algorithm for Problem \eqref{EquOriginal}} \vspace{-0.1cm}
We first tackle the non-smoothness of the objective function of problem \eqref{EquOriginal} by using the following lemma. \vspace{-0.1cm}
\begin{lemma}
Problem \eqref{EquOriginal} is equivalent\footnote{In this paper, the word ``equivalent'' means that both problems share the same optimal solution.} to the following problem:  \vspace{-0.2cm}
\begin{align}
 \max_{ \mathbf{x}, \mathbf{y}, \mathbf{P} } & \;  \sum_{n=1}^N  \bigg[ R_{\text{AB}}[n]  - \max_{k\in \mathcal{K}} \max_{ (\Delta x_k, \Delta y_k) \in \mathcal{E}_k } R_{\text{AE}_k}[n] \bigg]  \label{EquReform} \\
\text{s.t.} \; & \; \eqref{EquMobilityCon}, \; \eqref{EquPowerCon}. \nonumber  \vspace{-0.1cm}
\end{align}
\end{lemma}
\begin{proof}
Denote $W_1^*$ and $W_2^*$ as the optimal values of problems \eqref{EquOriginal} and \eqref{EquReform}, respectively. First, since $[x]^+ \geq x, \forall x$, we have $W_1^* \geq W_2^*$. Next, denote $(\mathbf{x}^*, \mathbf{y}^*, \mathbf{P}^*)$ as the optimal solution to \eqref{EquOriginal}, where $\mathbf{P}^*=[ P^*[1],\ldots,P^*[N] ]^\dagger$. Let $f(P[n]) = R_{\text{AB}}[n]  - \max_{k\in \mathcal{K}} \max_{ (\Delta x_k, \Delta y_k) \in \mathcal{E}_k } R_{\text{AE}_k}[n]$. We construct a feasible solution $(\hat{\mathbf{x}},\hat{\mathbf{y}},\hat{\mathbf{P}})$ to \eqref{EquReform}, such that $\hat{\mathbf{x}} = \mathbf{x}^*$, $\hat{\mathbf{y}} = \mathbf{y}^*$, and the elements of $\hat{\mathbf{P}}$ are obtained as: if $f(P^*[n])\geq 0$, $\hat{P}[n] = P^*[n]$; otherwise $\hat{P}[n] = 0$. Denote the objective value of \eqref{EquReform} attained at $(\hat{\mathbf{x}},\hat{\mathbf{y}},\hat{\mathbf{P}})$ as $\hat{W}$. The newly constructed solution $(\hat{\mathbf{x}},\hat{\mathbf{y}},\hat{\mathbf{P}})$ ensures that $\hat{W}=W_1^*$. Since $(\hat{\mathbf{x}},\hat{\mathbf{y}},\hat{\mathbf{P}})$ is feasible to \eqref{EquReform}, it follows that $W_2^* \geq \hat{W}$ and thus $W_2^* \geq W_1^*$. Therefore, $W_1^* = W_2^*$, which completes the proof.
\end{proof}

Although problem \eqref{EquReform} is more tractable, it is still difficult to solve due to its non-convexity. Nevertheless, we observe that the optimization variables can be partitioned into two blocks, i.e., $(\mathbf{x}, \mathbf{y})$ and $\mathbf{P}$, which facilitates the algorithm design for solving problem \eqref{EquReform} via the block coordinate descent method \cite{WuGC2017,Boyd2004}. Specifically, we solve \eqref{EquReform} by solving the following two sub-problems iteratively: sub-problem 1 optimizes $\mathbf{P}$ under given $(\mathbf{x}, \mathbf{y})$; while sub-problem 2 optimizes $(\mathbf{x}, \mathbf{y})$ under given $\mathbf{P}$, as detailed in the next two subsections, respectively. In the end, we summarize the overall algorithm and show its convergence.

\vspace{-0.4cm}

\subsection{Solution to Sub-Problem 1} \vspace{-0.1cm}
For given $(\mathbf{x}, \mathbf{y})$, sub-problem 1 can be written as  \vspace{-0.2cm}
\begin{equation} \label{EquSubProb1}
\max_{\mathbf{P}}  \; \sum_{n=1}^N \big[ \log_2(1+\alpha_n P[n]) - \log_2(1+\beta_n)  \big] \; \; \text{s.t.} \; \eqref{EquPowerCon},  \vspace{-0.1cm}
\end{equation}
where  \vspace{-0.2cm}
\begin{equation}  \label{Equa}
\alpha_n = \frac{ \gamma_0 } {  x^2[n] + y^2[n] + H^2 }, \; \;  \beta_n = \frac{ \gamma_0  } { \min_{k \in \mathcal{K}} \theta_{k,n} }, \quad \; \;  \vspace{-0.1cm}
\end{equation}  \vspace{-0.2cm}
\begin{equation}  \label{Equck}
\theta_{k,n} =  \min_{ (\Delta x_k, \Delta y_k) \in \mathcal{E}_k } (x[n]-x_k)^2 + (y[n]-y_k)^2 + H^2.  \vspace{-0.1cm}
\end{equation}
By substituting \eqref{EquLocErr} and \eqref{EquErrArea} into \eqref{Equck}, $\theta_{k,n}$ can be obtained as  \vspace{-0.2cm}
\begin{equation}
\theta_{k,n} = \begin{cases}  H^2  &  d_k \leq Q_k , \\  (d_k-Q_k)^2 + H^2 & d_k > Q_k,  \end{cases} \vspace{-0.05cm}
\end{equation}
where $d_k = \sqrt{ (x[n]-x_{\text{E}_k})^2 + ( y[n]-y_{\text{E}_k} )^2 }$. The optimal solution to problem \eqref{EquSubProb1} can be obtained as \cite{Gopala2008}  \vspace{-0.2cm}
\begin{equation}  \label{EquPowSol}
P^*[n] = \begin{cases} \min \left( [ \hat{P}[n] ]^+ , P_{\text{peak}} \right) & \alpha_n > \beta_n,  \\ 0 & \alpha_n \leq \beta_n , \end{cases} \vspace{-0.1cm}
\end{equation}
where  \vspace{-0.2cm}
\begin{equation}  \label{EquOptP}
\hat{P}[n] =  \sqrt{ \left( \frac{1}{2\beta_n} - \frac{1}{2\alpha_n} \right)^2 + \frac{1}{\lambda \ln 2} \left( \frac{1}{\beta_n} - \frac{1}{\alpha_n}  \right) } - \frac{1}{2\beta_n} - \frac{1}{2\alpha_n}.
\end{equation}
In \eqref{EquOptP}, $\lambda \geq 0$ is a parameter to ensure that the constraint \eqref{EquAvgPowCon} is satisfied at the optimal solution, which can be determined by bisection search \cite{Boyd2004}.

\vspace{-0.4cm}

\subsection{Solution to Sub-Problem 2} \vspace{-0.1cm}
 %single column equation
\begin{figure*}[!t]
% ensure that we have normalsize text
\normalsize
\begin{align}
\max_{\mathbf{x},\mathbf{y}}   \sum_{n=1}^N \bigg[ \log_2 \left( 1+\frac{P_n}{ x^[n] + y^2[n] + H^2 } \right)   - \log_2 \bigg( 1 + \frac{P_n}{ \min\limits_{k \in \mathcal{K}} \min\limits_{ (\Delta x_k, \Delta y_k) \in \mathcal{E}_k  }  (x[n]-x_k)^2 + (y[n]-y_k)^2 + H^2 }  \bigg)   \bigg] \; \text{s.t.} \eqref{EquMobilityCon}. \label{EquProbSub2}  \end{align}
% IEEE uses as a separator
\hrulefill
% The spacer can be tweaked to stop underfull vboxes.
\end{figure*}

By setting $P_n = \gamma_0 P[n]$, sub-problem 2 can be expressed as \eqref{EquProbSub2} shown at the top of next page, which cannot be solved optimally in polynomial time due to its non-convexity. By introducing slack variables $\mathbf{u}  \triangleq [u[1],\ldots,u[N]]^\dagger$ and $\mathbf{t}  \triangleq [t[1],\ldots,t[N]]^\dagger$, we first consider the following equivalent problem:  \vspace{-0.2cm}
\begin{subequations}  \label{EquReform2}
	\begin{align}
	\max_{ \mathbf{x}, \mathbf{y}, \mathbf{u},\mathbf{t} } & \;  \sum_{n=1}^N \bigg[  \log_2 \left( 1+ \frac{ P_n } { u[n]  }  \right) -  \log_2 \bigg( 1 + \frac{ P_n }{ t[n] } \bigg) \bigg] \label{EquReform2Obj}   \\
	\text{s.t.}  \; \; &   \min_{ (\Delta x_k, \Delta y_k) \in \mathcal{E}_k } (x[n]- x_k)^2 + (y[n] - y_k)^2 + H^2   \nonumber \\
	& \geq t[n], \; \forall n,k, \label{EquCont}   \\
	& x^2[n] + y^2[n] + H^2 - u[n] \leq 0, \; \forall n, \label{EquConu} \\
	&  t[n] \geq H^2, \; \forall n, \; \eqref{EquMobilityCon}.\nonumber
	\end{align}
\end{subequations}
{Problems \eqref{EquProbSub2} and \eqref{EquReform2} have the same optimal solution of $(\mathbf{x},\mathbf{y})$, since the constraints \eqref{EquCont} and \eqref{EquConu} are active at the optimal solution to problem \eqref{EquReform2}. This can be proved by contradiction: if constraints \eqref{EquCont} and \eqref{EquConu} are inactive, the objective value of \eqref{EquReform2} can be improved by increasing $t[n]$ (decreasing $u[n]$). Hence, we can focus on solving problem \eqref{EquReform2}.} However, problem \eqref{EquReform2} is still intractable, since there is an infinite number of $(\Delta x_k, \Delta y_k)$ in constraint \eqref{EquCont} due to the continuous nature of $\mathcal{E}_k$. Now, we convert \eqref{EquCont} into equivalent constraints as follows. First, we substitute \eqref{EquLocErr} and \eqref{EquErrArea} into \eqref{EquCont} and rewrite it as  \vspace{-0.2cm}
\begin{subequations}
	\begin{align}
	& \Delta x_k^2 + \Delta y_k^2 -Q_k^2 \leq 0 ,\; \forall k,  \label{EquInequ1}  \\
	& -(x[n]- x_{\text{E}_k} -\Delta x_k )^2 - (y[n] - y_{\text{E}_k} -\Delta y_k)^2 \nonumber \\
	&- H^2 + t[n] \leq 0 , \; \forall k. \label{EquInequ2} \vspace{-0.1cm}
	\end{align}
\end{subequations}
Next, according to $\mathcal{S}$-Procedure \cite{Boyd2004}, since there exists a point $(\Delta \hat{x}_k, \Delta \hat{y}_k)$ (e.g., $(\Delta \hat{x}_k, \Delta \hat{y}_k)=(0,0)$) such that $\Delta \hat{x}_k^2 + \Delta \hat{y}_k^2 -Q_k^2 < 0$, the implication \eqref{EquInequ1} $\Rightarrow$ \eqref{EquInequ2} holds if and only if there exists $\xi_k[n] \geq 0$ such that  \vspace{-0.15cm}
\begin{equation}   \label{EquSProcMat1}
\boldsymbol{\Phi}(x[n], y[n], t[n], \xi_k[n]) \succeq \mathbf{0}, \; \forall k,n, \vspace{-0.1cm}
\end{equation}
where  \vspace{-0.2cm}
\begin{displaymath}
\begin{split}
& \boldsymbol{\Phi}(x[n], y[n], t[n], \xi_k[n])  \\
\; \; \; \; \; \;  \; =& \begin{bmatrix}
\xi_k[n] + 1& 0 & x_{\text{E}_k}-x[n] \\ 0 & \xi_k[n]+1 & y_{\text{E}_k} - y[n]  \\ x_{\text{E}_k}-x[n]  & y_{\text{E}_k} - y[n] & -Q_k^2 \xi_k[n] + c_k[n]
\end{bmatrix}, \; \text{and} \vspace{-0.1cm}
\end{split}
\end{displaymath} \vspace{-0.2cm}
\begin{align}
c_k[n] =& \; x^2[n] - 2 x_{\text{E}_k} x[n] + x_{\text{E}_k}^2 + y^2[n] - 2 y_{\text{E}_k} y[n] + y_{\text{E}_k}^2  \; \; \; \; \; \nonumber \\
&   +H^2 - t[n].   \label{Equc}
\end{align}

By replacing \eqref{EquCont} with \eqref{EquSProcMat1} and introducing slack variables $\boldsymbol{\Xi} \triangleq [\boldsymbol{\xi}_1, \ldots, \boldsymbol{\xi}_K ]$, where $\boldsymbol{\xi}_k \triangleq [\xi_k[1],\ldots,\xi_k[N] ]^\dagger$, we rewrite problem \eqref{EquReform2} into an equivalent form:  \vspace{-0.2cm}
\begin{subequations}  \label{EquReform3}
	\begin{align}
	\max_{ \mathbf{x}, \mathbf{y}, \mathbf{u},\mathbf{t}, \boldsymbol{\Xi} } \; &   \sum_{n=1}^N \bigg[  \log_2 \left( 1+ \frac{ P_n } { u[n]  }  \right) -  \log_2 \bigg( 1 + \frac{ P_n }{ t[n] } \bigg) \bigg]  \label{EquReform3Obj} \\
	\text{s.t.}  \; \;\;  &   \boldsymbol{\Phi}(x[n], y[n], t[n], \xi_k[n]) \succeq \mathbf{0}, \; \forall k,n,   \label{EquReform3Con2}    \\
	& t[n] \geq H^2, \; \xi_k[n] \geq 0, \; \forall k,n,\label{EquReform3Con3}  \\
	& \eqref{EquMobilityCon}, \; \eqref{EquConu}. \nonumber
	\end{align}
\end{subequations}
The objective function in \eqref{EquReform3Obj} is non-concave, since $ \log_2 ( 1+ \frac{ P_n } { u[n]  } )$ is convex. Moreover, the constraint \eqref{EquReform3Con2} is non-convex, since the terms $x^2[n]$ and $y^2[n]$ contained in $c_k[n]$ (see \eqref{Equc}) are non-linear. Thus, problem \eqref{EquReform3} is difficult to be solved optimally due to its non-convexity. We propose an iterative algorithm to find an approximate solution to problem \eqref{EquReform3} as follows. First, the algorithm assumes a feasible point $\mathbf{x}_{\text{fea}} \triangleq \left[ x_{\text{fea}}[1], \ldots, x_{\text{fea}}[N] \right]^{\dagger}$, $\mathbf{y}_{\text{fea}} \triangleq \left[ y_{\text{fea}}[1], \ldots, y_{\text{fea}}[N] \right]^{\dagger}$ and $\mathbf{u}_{\text{fea}} \triangleq \left[ u_{\text{fea}}[1], \ldots, u_{\text{fea}}[N] \right]^{\dagger}$, which is feasible to \eqref{EquReform3}. Then, by using the first-order Taylor expansions of $\log_2 ( 1+ \frac{ P_n } { u[n]  } )$, $x^2[n]$ and $y^2[n]$ at $u_{\text{fea}}[n]$, $x_{\text{fea}}[n]$ and $y_{\text{fea}}[n]$, respectively,  \vspace{-0.2cm}
\begin{align}
\log_2 \left( 1+ \frac{ P_n } { u[n]  }  \right) \geq & \log_2 \left( 1+ \frac{ P_n } { u_{\text{fea}}[n]  }  \right)  \nonumber \\
& - \frac{ P_n ( u[n] - u_{\text{fea}}[n] ) }{ \ln 2 ( u_{\text{fea}}^2[n] + P_n u_{\text{fea}}[n] ) },  \label{EquTayloru}
\end{align}
\begin{equation}   \label{EquTaylorx}
x^2[n] \geq - x_{\text{fea}}^2[n] + 2 x_{\text{fea}}[n] x[n], \vspace{-0.1cm}
\end{equation}
\begin{equation}   \label{EquTaylory}
y^2[n] \geq - y_{\text{fea}}^2[n] + 2 y_{\text{fea}}[n] y[n], \vspace{-0.1cm}
\end{equation}
problem \eqref{EquReform3} is approximately transformed to \vspace{-0.2cm}
\begin{subequations}  \label{EquSubProb2Recast}
	\begin{align}
	\max_{ \mathbf{x}, \mathbf{y}, \mathbf{u},\mathbf{t}, \boldsymbol{\Xi} } & \;  \sum_{n=1}^N -\frac{ P_n ( u[n] - u_{\text{fea}}[n] ) }{  \ln2 (u_{\text{fea}}^2[n] + P_n u_{\text{fea}}[n] ) }  -  \log_2 \bigg( 1 + \frac{ P_n }{ t[n] } \bigg)  \\
	\text{s.t.}  \; \;\;  &  \tilde{ \boldsymbol{\Phi} } (x[n], y[n], t[n], \xi_k[n]) \succeq \mathbf{0} , \; \forall k,n,    \label{EquSubProb2RcCon2}   \\
	& \eqref{EquMobilityCon}, \; \eqref{EquConu}, \; \eqref{EquReform3Con3}. \nonumber
	\end{align}
\end{subequations}
where  \vspace{-0.2cm}
\begin{align}
& \tilde{ \boldsymbol{\Phi} } (x[n], y[n], t[n], \xi_k[n]) \nonumber \\
\quad \; \; \; =& \begin{bmatrix}
\xi_k[n] + 1& 0 & x_{\text{E}_k}-x[n] \\ 0 & \xi_k[n]+1 & y_{\text{E}_k}-y[n] \\ x_{\text{E}_k}-x[n] & y_{\text{E}_k}-y[n] & -Q_k^2 \xi_k[n] + \tilde{c}_k[n]
\end{bmatrix}, \; \text{and}  \nonumber
\end{align}
\begin{align}
 &\tilde{c}_k[n] = - x_{\text{fea}}^2[n] + 2 x_{\text{fea}}[n] x[n] - 2 x_{\text{E}_k} x[n] + x_{\text{E}_k}^2 - y_{\text{fea}}^2[n]  \; \;  \; \nonumber \\
& \quad  \quad \; + 2 y_{\text{fea}}[n] y[n] - 2 y_{\text{E}_k} y[n] + y_{\text{E}_k}^2 +H^2 - t[n].  \label{Equctilde}
\end{align}
Note that problem \eqref{EquSubProb2Recast} is a semidefinite programming problem, which can be optimally solved by the interior-point method \cite{Boyd2004}.

{ \emph{Remark 1:} Since \eqref{EquTaylorx} and \eqref{EquTaylory} are lower bounds for $x^2[n]$ and $y^2[n]$, respectively \cite{Boyd2004}, we have $c_k[n] \geq \tilde{c}_k[n]$ and thus \vspace{-0.2cm}
\begin{equation}  \label{EquLargerPSM}
\boldsymbol{\Phi}  (x[n], y[n], t[n], \xi_k[n]) \succeq \tilde{ \boldsymbol{\Phi} } (x[n], y[n], t[n], \xi_k[n]), \vspace{-0.1cm}
\end{equation}
which means that \eqref{EquSubProb2RcCon2} implies \eqref{EquReform3Con2}. Hence, the solution to problem \eqref{EquSubProb2Recast} must be a feasible solution to problem \eqref{EquReform3}.

\emph{Remark 2:} Since \eqref{EquTayloru} is a lower bound of $\log_2 ( 1+ \frac{ P_n } { u[n]  } )$ \cite{Boyd2004}, problem \eqref{EquSubProb2Recast} maximizes a lower bound of the objective function of \eqref{EquReform3}. This lower bound is equal to the objective value of \eqref{EquReform3} only at $(\mathbf{x}_{\text{fea}}, \mathbf{y}_{\text{fea}}, \mathbf{u}_{\text{fea}})$, so the objective value of problem \eqref{EquReform3} with the solution to problem \eqref{EquSubProb2Recast} is equal to or greater than that with the solution $(\mathbf{x}_{\text{fea}}, \mathbf{y}_{\text{fea}}, \mathbf{u}_{\text{fea}})$.}

\vspace{-0.5cm}

\subsection{Overall Algorithm}  \label{SecOverallAlg}  \vspace{-0.1cm}
{We summarize the detail of the overall algorithm in Algorithm 1, which solve problems \eqref{EquSubProb1} and \eqref{EquSubProb2Recast} alternatively and iteratively until it converges. Since as shown in the previous two subsections, the objective value of problem \eqref{EquReform} with the solutions obtained by solving problems \eqref{EquSubProb1} and \eqref{EquSubProb2Recast} is non-decreasing over iterations and the optimal value of problem \eqref{EquReform} must be finite, the solution obtained by Algorithm 1 can be guaranteed to converge to a suboptimal solution \cite{WuGC2017}. Algorithm 1 is suitable for UAV applications, since it has a complexity of $\mathcal{O} \left[ N_{\text{ite}} (4N+KN)^{3.5} \right]$ and can obtain the solution in polynomial time, where $N_{\text{ite}}$ is the iteration number \cite{Boyd2004}. }
\begin{algorithm}[!t]
	\caption{Proposed Algorithm for Problem \eqref{EquOriginal}.}
	\begin{algorithmic}[1] {
		\STATE Initialize $\mathbf{P}^{(0)}$, $\mathbf{x}^{(0)}$, $\mathbf{y}^{(0)}$, and $\mathbf{u}^{(0)}$. Set  $m=0$.		
		\REPEAT
		\STATE Set $m \gets m+1$.	
		\STATE Let $\mathbf{x}_{\text{fea}} = \mathbf{x}^{(m-1)}$, $\mathbf{y}_{\text{fea}} = \mathbf{y}^{(m-1)}$ and $\mathbf{u}_{\text{fea}} = \mathbf{u}^{(m-1)}$. Solve problem \eqref{EquSubProb2Recast} under given $\mathbf{P}^{(m-1)}$ to obtain $(\mathbf{x}^{(m)}, \mathbf{y}^{(m)})$.
		\STATE Solve problem \eqref{EquSubProb1} under given $(\mathbf{x}^{(m)}, \mathbf{y}^{(m)} )$ to obtain $\mathbf{P}^{(m)}$.		
		\UNTIL {The fractional increase of the objective value is smaller than a threshold $\epsilon>0$.} }
	\end{algorithmic}
\end{algorithm}

\begin{figure*}[!t]
\centering
\subfloat[Trajectories of Alice.]{\includegraphics[width=0.31\textwidth]{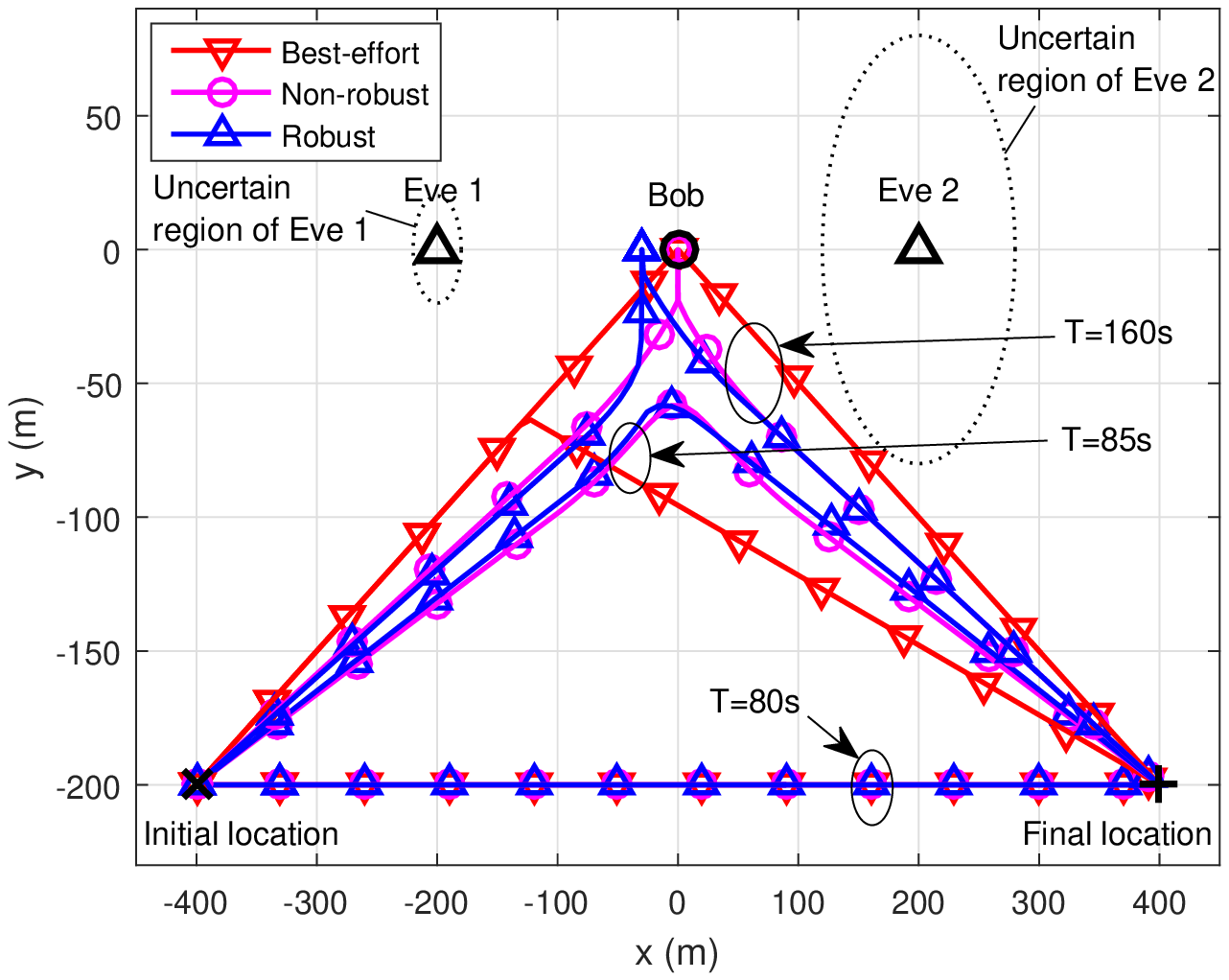}
\label{FigTraj}}
\hfil
\subfloat[Average secrecy rate versus $T$.]{\includegraphics[width=0.31\textwidth]{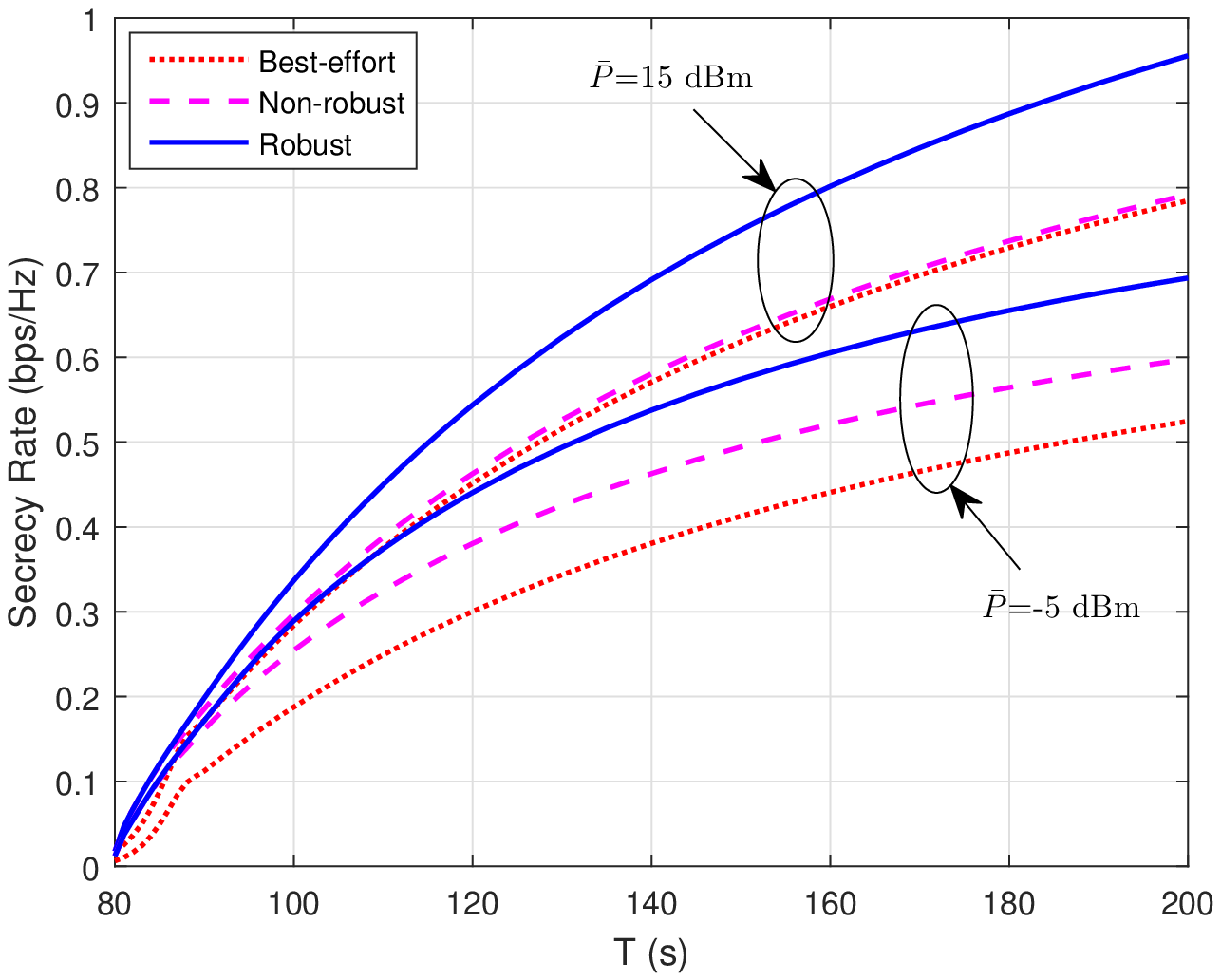}
\label{FigSRvsT}}
\hfil
\subfloat[Average secrecy rate versus $\bar{P}$.]{\includegraphics[width=0.31\textwidth]{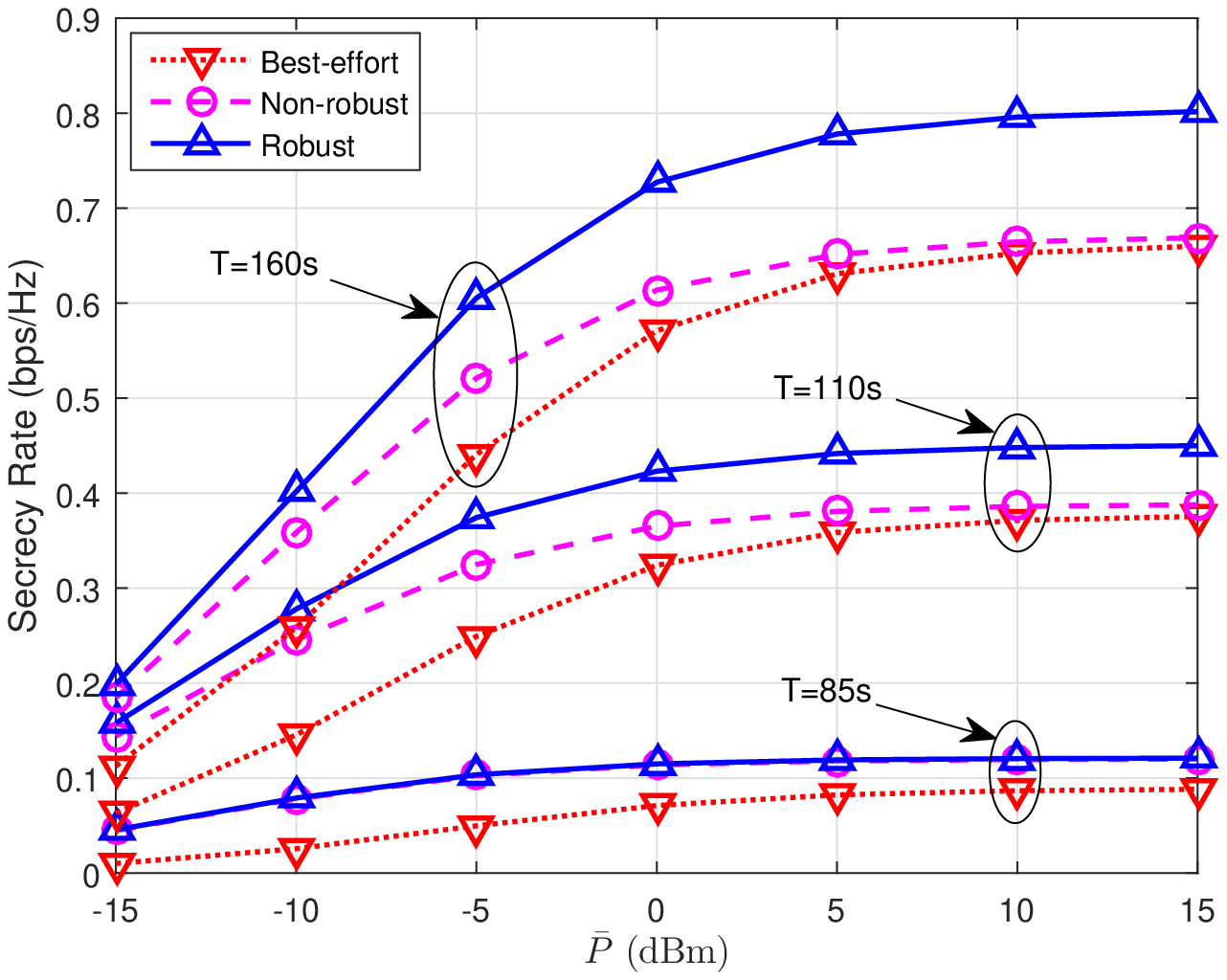} \label{FigSRvsP}}
\caption{Trajectories and average worst-case secrecy rates of different algorithms.}
\label{fig_sim}
\end{figure*}

\vspace{-0.2cm}
\section{Simulation Results}  \vspace{-0.05cm}
{This section provides simulation results to verify the performance of our proposed robust joint trajectory and transmit power design algorithm, as compared to the following two benchmark algorithms: 1) non-robust joint trajectory and transmit power design; 2) best-effort trajectory design with equal transmit power \cite{Zhang2017GC}}. Specifically, the non-robust algorithm only has the estimated locations of Eves and treats them as perfect information. Thus, it jointly designs the UAV trajectory and transmit power control by using Algorithm 1 assuming $Q_k=0$, $\forall k$. {The best-effort algorithm performs equal transmit power allocation over time and designs Alice's trajectory in the following heuristic best-effort manner: Alice flies to the location right above Bob at speed $v_{\max}$, then hovers there, and finally flies at speed $v_{\max}$ to reach the final location at time $T$. If Alice does not have sufficient time to reach the location above Bob, it will turn its direction midway and fly to the final location directly.} The initial feasible points for the proposed robust and benchmark non-robust algorithms are generated by the best-effort algorithm. There are $K=2$ Eves, whose estimated horizontal coordinates are $(x_{\text{E},1}, y_{\text{E},1}) = (-200,0)$m and $(x_{\text{E},2}, y_{\text{E},2}) = (200,0)$m, respectively, and $Q_1=20$m and $Q_2=80$m. The other parameters are set as follows: $H=100$m, $v_{\max}=10$m/s, $d_t=0.5$s, $\gamma_0=80$dB, $P_{\text{peak}} =4 \bar{P}$, $(x[0],y[0])=(-400,-200)$m, $(x[N],y[N])=(400,-200)$m, and $\epsilon=10^{-4}$.

Fig. 2(a) shows the trajectories of Alice by applying different algorithms when $\bar{P}=-5$dBm. It is observed that when $T$ is small (e.g., $T=80$s), the trajectories obtained by the robust and non-robust algorithms are very similar, since the flexibility in trajectory design is limited as Alice is required to fly from the initial location to the final location in a given duration $T$. As $T$ increases, the flexibility in designing efficient trajectory increases. This magnifies the differences between the robust and non-robust algorithms. When $T$ is sufficiently large (e.g., $T=160$s), by these two algorithms, Alice first flies at its maximum speed in an arc path to keep away from Eve 1 and reaches a certain point near Bob; then it hovers at that point as long as possible, and finally flies to the final location along an arc path bypassing Eve 2 at its maximum speed. However, in the proposed robust algorithm, the hovering point is on the left of Bob; while in the non-robust algorithm, the hovering point is directly above Bob. This is because although the distances from the estimated locations of Eves 1 and 2 to Bob are equal, the radius of the uncertain region of Eve 1 is smaller than that of Eve 2. {Considering the worst case, the proposed robust algorithm adjusts the hovering point closer to Eve 1 and farther from Eve 2 to strike a balance between enhancing the legitimate link from Alice to Bob and degrading the quality of the eavesdropping links from Alice to Eves 1 and 2, while the non-robust algorithm fails to strike such a balance by treating Eve 1 and Eve 2 equally.}

Figs. 2(b) and 2(c) show the corresponding average worst-case secrecy rates of all algorithms versus the flight duration $T$ and average power $\bar{P}$, respectively. In both figures, it can be observed that the secrecy rates of all algorithms increase with $T$ and $\bar{P}$. In particular, the proposed robust algorithm significantly outperforms the other two benchmark algorithms. In Fig. 2(c), it is observed that the secrecy rates of all algorithms are saturated when $\bar{P}$ is high. This is because as shown in \eqref{EquRAB}--\eqref{EquSecrecyRate}, the secrecy rate maximization problem \eqref{EquOriginal} is independent on the transmit power $P[n]$ and only depends on the UAV trajectory in the high transmit power regime. Furthermore, Fig. 2(c) shows that although the non-robust algorithm outperforms the best-effort algorithm, the secrecy rate gap between them becomes smaller as $\bar{P}$ increases. This is because the non-robust algorithm ignores the location estimation errors of Eves 1 and 2, and thus suffers from the performance loss. The above results demonstrate the importance and the potential performance gain brought by the robust joint trajectory and transmit power design.

\vspace{-0.1cm}

\section{Conclusion} \vspace{-0.05cm}
This paper investigated a secure UAV communication system when the locations of the eavesdroppers are not perfectly known as in the practical scenario. A robust joint trajectory and transmit power design algorithm was proposed to maximize the average worst-case secrecy rate subject to the UAV's mobility constraints as well as its average and peak transmit power constraints. Simulation results showed that the proposed joint design algorithm which considers the location uncertainties of Eves can improve the worst-case secrecy rate performance significantly, as compared to two benchmark algorithms without considering the uncertainties of the eavesdroppers' location information.

% that's all folks
\end{document}